\documentclass[10pt]{article} 
\usepackage{latexsym}

\newtheorem{definition}{Definition} 
\newtheorem{example}{Example} 
\newtheorem{theorem}{Theorem} 
\newtheorem{lemma}{Lemma} 
\newtheorem{proof}{Proof} 

\begin{document}
\bibliographystyle{acmtrans}

\title{Integrating Interval Constraints\\
       into Logic Programming}

\date{}

\author{M.H. van Emden\\
        Department of Computer Science\\
        University of Victoria \\
        Research report DCS-133-IR
       }


\newcommand{\cpp}{\hbox{{\tt C++}}}

\newcommand{\Rat}{\ensuremath{\mbox{\textbf{Q}}}} 
\newcommand{\Rea}{\ensuremath{\mathcal{R}}} 
\newcommand{\ExtRe}{\ensuremath{\mbox{\textbf{R}}^{++}}} 
\newcommand{\Flpt}{\ensuremath{\mbox{\textbf{F}}}} 

\newcommand{\A}{\ensuremath{\mathcal{A}}} 
\newcommand{\Bool}{\ensuremath{\mathcal{B}}} 
\newcommand{\C}{\ensuremath{\mathcal{C}}} 
\newcommand{\D}{\ensuremath{\textbf{D}}} 
\newcommand{\DE}{\ensuremath{\mathcal{DE}}} 
\newcommand{\E}{\ensuremath{\mathcal{E}}} 
\newcommand{\F}{\ensuremath{\textbf{F}}} 
\newcommand{\Herb}{\ensuremath{\mathcal{H}}} 
\newcommand{\I}{\ensuremath{\mathcal{I}}} 
\newcommand{\Nat}{\ensuremath{\mathcal{N}}} 
\newcommand{\M}{\ensuremath{\textbf{M}}} 
\newcommand{\MI}{\ensuremath{\mathcal{M}_I}} 
\newcommand{\Pred}{\ensuremath{\textbf{P}}} 
\newcommand{\R}{\ensuremath{\mathcal{R}}} 
\newcommand{\Rz}{\ensuremath{\textbf{R}}} 
\newcommand{\Sch}{\ensuremath{\mathcal{S}}} 
\newcommand{\T}{\ensuremath{\textbf{T}}} 
\newcommand{\Var}{\ensuremath{\textbf{V}}} 
\newcommand{\Vars}{\mathcal{X}}
\newcommand{\Z}{\ensuremath{\mathcal{Z}}} 

\newcommand{\cart}{\ensuremath{\mbox{\textsc{cart}}}} 
\newcommand{\apl}{\ensuremath{\mbox{\textsc{apl}}}} 
\newcommand{\real}{\ensuremath{\mbox{\textsc{real}}}} 
\newcommand{\preal}{\ensuremath{\mbox{\textsc{preal}}}} 
\newcommand{\tit}{\ensuremath{\mbox{\textit{true}}}} 
\newcommand{\tbf}{\ensuremath{\mbox{\textbf{true}}}} 
\newcommand{\fit}{\ensuremath{\mbox{\textit{false}}}} 
\newcommand{\fbf}{\ensuremath{\mbox{\textbf{false}}}} 

\newcommand{\va}{{\tt a}}
\newcommand{\vb}{{\tt b}}
\newcommand{\vfa}{{\tt f(a)}}
\newcommand{\vfb}{{\tt f(b)}}

\newcommand{\pair}[2]{\ensuremath{\langle #1,#2 \rangle}}
\newcommand{\triple}[3]{\ensuremath{\langle #1,#2,#3 \rangle}}
\newcommand{\vc}[2]{\ensuremath{#1_0,\ldots,#1_{#2-1}}}

\newcommand{\bnr}{{\scshape bnr}}
\newcommand{\chip}{{\scshape chip}}
\newcommand{\ichip}{{\scshape ichip}}
\newcommand{\pr}{{\scshape clp/ncsp}}
\newcommand{\clpdc}{{\scshape clp/dc}}
\newcommand{\clpncsp}{{\scshape clp/ncsp}}
\newcommand{\clp}{{\scshape clp}}
\newcommand{\clpr}{{\scshape clp(r)}}
\newcommand{\csp}{{\scshape csp}}
\newcommand{\csps}{{\scshape csp}s}
\newcommand{\dc}{{\scshape dc}}
\newcommand{\sld}{{\scshape sld}}
\newcommand{\re}{{\ttfamily re}}
\newcommand{\dec}{{\ttfamily dec}}
\newcommand{\fpi}{{\ttfamily fpi}}
\newcommand{\intv}[2]{{\ttfamily [}#1,#2{\ttfamily ]}}

\newcommand{\sol}{\ensuremath{\mbox{\textsc{sol}}}}

\newcommand{\ex}{\ensuremath{\mbox{Ex}}}

\newcommand{\para}{\vspace{0.05in}}

\newcommand{\solved}{\ensuremath{\mbox{{\tt solved}}}}
\newcommand{\solvedd}{\ensuremath{\mbox{{\tt solved1}}}}
\newcommand{\val}{\ensuremath{\mbox{{\tt val}}}}
\newcommand{\inn}{\ensuremath{\mbox{{\tt in}}}}

\newcommand{\subr}{\rightarrow_r}
\newcommand{\subi}{\rightarrow_i}
\newcommand{\subc}{\rightarrow_c}
\newcommand{\subs}{\rightarrow_s}
\newcommand{\fail}{\hbox{ {\it fail} }}
\newcommand{\infer}{\hbox{{\it infer}}}
\newcommand{\consistent}{\hbox{{\it consistent}}}
\newcommand{\ap}{\hbox{{\it ap}}}
\newcommand{\app}{\hbox{{\it ap}}^{\prime}}

\newcommand{\df}{\mathrm{def}}
\newcommand{\abs}{\mathrm{abs}}
\newcommand{\dro}{{\scshape dro}}
\newcommand{\dros}{{\scshape dro}s}

\maketitle

\abstract{
The \clp\ scheme uses Horn clauses and \sld\ resolution to generate
multiple constraint satisfaction problems (\csps).
The possible \csps\ include rational trees (giving Prolog)
and numerical algorithms for solving
linear equations and linear programs (giving \clpr).
In this paper we develop a form of \csp\ for interval constraints.
In this way one obtains a logic semantics
for the efficient floating-point hardware
that is available on most computers.

The need for the method arises
because in the practice of scheduling and engineering design
it is not enough to solve a single \csp.
Ideally one should be able to consider thousands of \csps\
and efficiently solve them or show them to be unsolvable.
This is what \pr, the new subscheme of \clp\ described in this paper
is designed to do.
}

\section{Introduction}

Floating-point arithmetic is marvelously cheap,
and it works most of the time.
Many textbooks on numerical analysis
contain examples of how spectacularly, or insidiously,
it can go wrong when it does \emph{not} work.
It would seem that in a mature computing technology
there is only place for reliable techniques.
Yet floating-point arithmetic is not to be lightly dismissed:
it is one of the main beneficiaries
of the enormous increase in processor performance
of the last few decades.
In combination with the insatiable demand
for more computationally intensive mathematical modeling,
this gives every motivation to use
interval arithmetic \cite{mrkrcl09}
as a way to safely use the dangerous technology
that is floating-point arithmetic.

Interval arithmetic ensures that,
in spite of the errors inherent in floating-point arithmetic,
a computation can be interpreted as a proof
that the real-valued result is contained
in the (interval) result of the computation.
However, the correctness provided by interval arithmetic
is limited to the evaluation of a single expression;
it does not extend to the algorithm
in which such evaluations take place.
To ensure correctness of the way an algorithm
combines expression evaluations one could of course use
verification methods for imperative programs,
such as Floyd's assertions.
In this paper we consider the alternative
of replacing the algorithms by logic programs,
thus allowing programs to be read as specifications.

Logic programming is more
than just an alternative to Floyd's assertions.
The logic framework suggests
a relational form for interval computation.
Such a relational form is provided
by interval \emph{constraints} \cite{bnldr97,bmcv94},
an improvement to interval arithmetic itself.
Incorporating interval constraints into logic programming
has the added advantage that the result goes beyond
the constraint processing paradigm by yielding programs
that \emph{generate}
multiple constraint satisfaction problems
in addition to solving them.
In scheduling and in engineering design it is typically the case
that one has an entire search space of such problems.
\pr, the integrated system described in this paper,
generates such search spaces.
Solving is not only used for obtaining results,
but also for pruning the search spaces by inducing early failure.

In Section~\ref{review} we start at the logic end
with a review of the \clp\ scheme \cite{jls87,jaffar-survey}. 
We use Clark's method \cite{clrk91}
for the semantics of logic programming schemes.
As this method uses a mild form of algebraic logic,
it needs some introduction;
this happens in Section~\ref{logPrel}.
In Section~\ref{sec:csp} we start at the opposite end
with a suitably modified version of the main features
of the Constraint Processing framework (\csp).
To bring together the two established
constraint approaches of the literature
we develop in Section~\ref{domain-constraints}
what we call here the \dc\ subscheme of the \clp\ scheme.
The integration of interval constraints
(reviewed in Section~\ref{sec:intervConstr})
into logic programming is described in Section~\ref{sec:clpncsp}.

\section{Related work}

The pioneering work in constraint logic programming is \cite{vnh89},
implemented as \chip\ \cite{dvhsagb88}.
Prolog has \sld\ resolution as sole inference rule;
\cite{vnh89} added Forward Checking, Look-Ahead,
and Partial Look-Ahead as additional inference rules,
to be applied to goals, depending how they are declared.

\chip\ was restricted to finite domains.
\ichip\ \cite{leevnmd93} proposed extending \chip\
to include floating-point intervals as domains
for real-valued variables.
Descendants of \chip\, such as the Eclipse system
(see \cite{ptwllc06} for a recent description),
implemented floating-point intervals.

The earliest design for integrating interval arithmetic
into Prolog is Cleary's \cite{clr87},
which served as basis for \bnr-Prolog \cite{BNR88,bnldr97}.
Cleary's proposal of a ``logical arithmetic'' for Prolog
described an implementation, but not a logical semantics.
His paper and \cite{dvs87} are the first to describe relational,
rather than functional,
interval arithmetic.
It remains to be seen whether the mathematical model
given by Older and Vellino \cite{ldrvll93}
can be connected to logic.
\bnr\ Prolog and Prolog IV \cite{colm96} are mentioned here
because of their connection with Prolog,
but not because of connection with logic programming.

The \clp\ scheme \cite{jls87} gives a logical semantics
that combines pure Prolog with constraint solving.
This scheme supersedes \chip\ and its descendants
as it is both simpler and more general.

The \clp\ scheme served as the basis
for the \clpr\ system \cite{jmsy92}.
It uses the scheme to generate answers
to numerical problems in the form of ``active constraints''.
In the derivation of these,
floating-point arithmetic is used without due precaution,
so that the validity of answers
is lost through rounding errors.

In \cite{vnmd97} it was shown
that the \clp\ scheme is general enough to
accommodate both interval and finite-domain constraints.
It does this by introducing ``value constraints''
without suggesting any way of interfacing these with
intervals.
This is done in this paper by means by means of \clpdc,
the \dc\ subscheme of the \clp\ scheme.
In this way we obtain \clpncsp, the first logic programming language
(as distinct from extension of Prolog)
with real variables in which only the precision,
but not the validity, of answers is affected by rounding errors.

\section{Logic Preliminaries} \label{logPrel}

\subsection{Relations}
Relations play a central role
in the integration of interval constraints into logic programming:
both constraints and the meanings of logic predicates
are relations.
Here we do not attempt to define
relations as generally as possible:
we only strive for adequacy for the purpose of this paper.
For a more drastic generalization of the usual notion
of relation, see \cite{vnmdSTPCS}.

As usually defined,
a relation is a subset of a Cartesian product
$S_1 \times \cdots \times S_k$.
That is, it consists of tuples
$\langle a_1, \ldots, a_k \rangle$
with $a_i \in S_i$ for $i = 1,\ldots,k$.
Such tuples are indexed
by the integers $1,\ldots, k$.
In the following we will need such relations,
as for example the ternary relation
$\mbox{sum } = \{\langle x,y,z \rangle \in \R^3 \mid x+y = z \}$,
indexed by the set $\{1,2,3\}$.

But we will also need relations
consisting of tuples indexed
by variables instead of integers.
For example, the constraint written
as $\mbox{sum}(x_2,x_2,x_1)$
is intended to be a relation distinct for the sum relation
just mentioned.
As another example,
$\mbox{sum}(x_2,x_2,x_1) \wedge \mbox{sum}(x_3,x_4,x_1)$
is intended to be a relation.
If so, which set of tuples? How indexed?

In this section we introduce the suitable type of relation;
in Section~\ref{subsec:denotations} we define
how they arise as meaning of the constraint expressions just shown.

\begin{definition}\label{def:relOn}
Given a set $\Vars = \{x_1, \ldots, x_N\}$ of variables,
a relation $\rho \subset S_1\times \cdots \times S_n$
consisting of tuples indexed by $\{1,\ldots,n\}$
and a sequence $\langle v_1, \ldots, v_n\rangle $ of variables
(not necessarily distinct),
the \underline{relation $\rho$ on $v$}
is the set of tuples $\tau$
indexed by the set of the $k \leq n$ variables in $v$
such that $\tau(v_i) = t_i$ for a tuple $t \in \rho$,
for all $i = 1,\ldots,n$. 
\end{definition}

\begin{example}
Let $\rho$ be the ternary sum relation
over the set $\mathcal{N}$ of natural numbers,
$\Vars = \{x_1, \ldots, x_{100}\}$, and $v = \langle x_2,x_2,x_1 \rangle$.
Then we have as example of a tuple $\tau$
in the relation $\rho$ on $v$:
\begin{eqnarray*}
\tau(v_1) = & \tau(x_2) & = t_1   \\
\tau(v_2) = & \tau(x_2) & = t_2   \\
\tau(v_3) = & \tau(x_1) & = t_3
\end{eqnarray*}
For such a $\tau$ to exist,
only tuples $t \in \rho$ qualify
where the first two elements are equal to each other.
$\tau$ consists of tuples indexed by the set $\{x_1, x_2\}$.
The tabular form of $\tau$ is as shown in Table~\ref{tab:taurelation}.
\begin{table} 
\begin{tabular}{r|r|r||r|r|r|r|r}
$\tau$  &  $x_1$ & $x_2$ & $\rho_1 \Join \rho_2$
                            &$x_1$ & $x_2$ & $x_3$ & $x_4$   \\
\hline 
        &   0    &  0 &     &0    &  0    &  0    &  0    \\
        &   2    &  1 &     &2    &  1    &  0    &  2   \\
        &   4    &  2 &     &2    &  1    &  1    &  1    \\
        &   6    &  3 &     &2    &  1    &  2    &  0   \\
        &   8    &  4 &     &4    &  2    &  0    &  4 \\
        &  10    &  5 &     &4    &  2    &  1    &  3 \\
        &\ldots&\ldots&     &\ldots&\ldots &\ldots &  \ldots    \\
\end{tabular}
\caption{On the left, tabular form of the relation $sum$ on
$\langle x_2,x_2,x_1,\rangle$ where
$sum = \{\langle x,y,z \rangle \in \R^3 \mid x+y = z \}$.
On the right,
tabular form of $\rho_1 \Join \rho_2$ from Example~\ref{ex:join}.
\label{tab:taurelation} 
}
\end{table}
\end{example}

\begin{definition}\label{def:join}
Let $\rho_1$ ($\rho_2$) be a relation
in which the tuples are indexed by
a set $\Vars_1$ ($\Vars_2$) of variables.
The \underline{join} of
$\rho_1$  and $\rho_2$ is
written as $\rho_1 \Join\ \rho_2$
and is defined as the relation
in which the tuples $\tau$ are indexed by 
$\Vars_1 \cup \Vars_2$ and are such that there exists,
for each tuple
$\tau \in (\rho_1 \Join\ \rho_2)$,
tuples
$\tau_1 \in \rho_1$
and
$\tau_2 \in \rho_2$
exist such that
$\tau(x) = \tau_1(x)$ if $x \in \Vars_1$
and
$\tau(x) = \tau_2(x)$ if $x \in \Vars_2$.
(Note that this implies that
$\tau_1$ and $\tau_2$
have to be such that $\tau_1(x) = \tau_2(x)$
for all $x$ such that $x \in \Vars1 \cap \Vars_2$.)
\end{definition}

\begin{example}\label{ex:join}
Let $\rho$ be the ternary sum relation
over the set $\mathcal{N}$ of natural numbers,
$\Vars = \{x_1, \ldots, x_{100}\}$,
$v_1 = \langle x_2,x_2,x_1 \rangle$, and
$v_2 = \langle x_3,x_4,x_1 \rangle$.
Let 
$\rho_1$ be $\rho$ on $v_1$
and
$\rho_2$ be $\rho$ on $v_2$.
Then $\rho_1 \Join\ \rho_2$ is a relation of which some tuples
are shown in Table~\ref{tab:taurelation}.
\end{example}

\subsection{Language}

The vocabulary of logic is formalized as a \emph{signature}
$\Sigma = \langle P,F,V \rangle$,
a tuple of disjoint, countably infinite, sets of
\emph{predicates},
\emph{functors}, and
\emph{variables}.
$P$ is partitioned according to whether
it may occur in a constraint or in a program.
Thus we have ``constraint predicates'' and
``program predicates''.
The constraint predicates include the nullary
\emph{true} and
\emph{false} and the binary $=$.

A \emph{term} is a variable or an expression
of the form $f(t_0,\ldots,t_{k-1})$,
where $f \in F$ and $t_0,\ldots,t_{k-1}$ are terms.
If $k = 0$, then the term is a \emph{constant}.

An \emph{atom} (or \emph{atomic formula})
is an expression of the form
$p(t_0,\ldots,t_{k-1})$,
where $p \in P$ is a predicate and
$t_0,\ldots,t_{k-1}$ are terms.
If $p$ is a program (constraint) predicate,
then an atom with $p$ as predicate is
a program (constraint) atom.

A goal statement is a conjunction of program atoms or constraint atoms.
A constraint is a conjunction of constraint atoms.

\subsection{Interpretations}
Interpretations depend on a language's signature.
They are formalized as $\Sigma$-\emph{structures}
$ \mathcal{I} = \langle \D,\Pred,\F \rangle $ where
\begin{itemize}
\item
\D\ is a non-empty set called the \emph{domain} of the interpretation.
\item
\Pred\ is a function mapping every $k$-ary predicate in $P$
to a subset of $\D^k$.
\Pred\ maps \tit\ to $\tbf = \{\langle \rangle\}$,
\fit\ to $\fbf = \{\}$, and
$=$ to $\{\pair{a}{a}\mid a \in \D\}$.
\item
\F\ is a function mapping every functor $f$ in $F$
to a function mapping each $k$-ary functor in $F$
to a k-adic function in $\D^k \rightarrow \D$.
\end{itemize}

\subsection{Denotations}\label{subsec:denotations}
An interpretation $ \langle \D,\F,\Pred \rangle $
determines a function \M\ mapping variable-free terms
to their denotations, as follows:
\begin{itemize}
\item
$\M(t) = \F(t) \in \D$ if $t$ is a constant
\item
$\M(f(t_0,\ldots,t_{k-1})) = (\F(f))(\M(t_0),\ldots,\M(t_{k-1}))$
for $k > 0$
\item
A ground atom
$p(t_0,\ldots,t_{k-1})$
is \tbf\ in an interpretation iff \\
$\langle \M(t_0),\ldots,\M(t_{k-1})\rangle \in \Pred(p) $
We give denotations of non-atomic formulas later, via relations.
\end{itemize}

We now consider denotations of terms and atoms that contain variables.
Let \A\ be an \emph{assignment},
which is a function in $V \rightarrow \D$,
assigning an individual in $\D$ to every variable.
(In other words, \A\ is a tuple of elements of $\D$
indexed by $V$).
As denotations of formulas with free variables
depend on \A\ , we write $\M_\A$.

\begin{itemize}
\item
$\M_\A(t) = \A(t)$ if $t \in V$
\item
$\M_\A(f(t_0,\ldots,t_{k-1})) =\\
      (\F(f))(\M_\A(t_0),\ldots,\M_\A(t_{k-1}))$
for $k > 0$.
\item
An atom
$p(t_0,\ldots,t_{k-1})$
is \tbf\ in an interpretation iff
$$ \langle \M_\A(t_0),\ldots,\M_\A(t_{k-1})\rangle
   \in \Pred(p)
$$
We give denotations of non-atomic formulas later, via relations.
\end{itemize}
The existential closure
$\exists x_0,\ldots,x_{n-1}$
of a set $C$ of atoms is true
in an interpretation
iff there is an assignment $\A$
such that $\M_\A(A) = \tbf$
for every atom $A \in C$.

\begin{definition}\label{def:bareDenot}
Let $\Vars_C \subset V$ be the set of the free variables in formula $C$. 
$R(C)$, the \emph{relation denoted by} $C$,
given the interpretation determining $\M_\A$, is defined as
$$
R(C) =  \{  \A \downarrow \Vars_C \mid
        \A \mbox{ is an assignment and }
        \M_\A(C) = \tbf
\}.
$$
\end{definition}
By $\A \downarrow \Vars_C$ we mean the function
$\A: V \rightarrow \D$ restricted to arguments in $\Vars_C \subset V$.

Thus $R(C)$ consists of tuples indexed by variables.
$R$ allows us to translate between
algebraic expressions in terms of relations
and formulas of logic.
This is useful because the results in
constraint satisfaction problems are expressed
in terms of relations,
whereas the constraint logic programming scheme
is expressed in terms of first-order predicate logic.

We have of course
$R(\tit) = \M(\tit) = \tbf$;
also
$R(\fit) = \M(\fit) = \fbf$.
More interestingly,
we may have
$R(C_1 \wedge C_2) = R(C_1) \cap R(C_2)$
and
$R(C_1 \vee C_2) = R(C_1) \cup R(C_2)$.
But these hold only when $C_1$ and $C_2$
have the same set of variables.
As this is not always the case,
we also need to define $R(Z,C)$,
where $Z$ is a set $\{z_1, \ldots z_N\}$ 
of variables \emph{containing}
$\Vars_C$, the set of the free variables of $C$:
\begin{definition}\label{def:looseDenot}
$$ R(Z,C) = \{\A \downarrow Z
      \mid \A \mbox{ is an assignment and } \M_\A(C) = \tbf
            \}.
$$
\end{definition}
Definitions
(\ref{def:bareDenot})
and
(\ref{def:looseDenot})
were suggested by a similar device first brought to our attention by
\cite{clrk91}.
The version here is modified to allow translations of a wider class of
formulas.
Their advantage is that of simplicity compared to other systems
of algebraic logic such as \cite{hmt85}.

\begin{eqnarray*}
R(\Vars_C,C)         &=& R(C)                    \\
R(\Vars_{C_1} \cup \Vars_{C_2}, C_1 \wedge C_2)
     &=& R(\Vars_{C_1} \cup \Vars_{C_2},C_1)
         \cap
         R(\Vars_{C_1} \cup \Vars_{C_2},C_2)  \\
R(\Vars_{C_1} \cup \Vars_{C_2}, C_1 \vee C_2)
      &=& R(\Vars_{C_1} \cup \Vars_{C_2},C_1)
         \cup
          R(\Vars_{C_1} \cup \Vars_{C_2},C_2).
\end{eqnarray*}
To be able to interface the CLP scheme, expressed in terms of predicate
logic formulas with CSPs, expressed in terms of relations,
we will use the following lemma.
\begin{lemma}
$R(\Vars_{C_1} \cup \Vars_{C_2}, C_1 \wedge C_2) = R(C_1) \Join R(C_2)$.
\end{lemma}

\subsection{Logical implication}

In the usual formulation of first-order predicate logic
we find the notation $T \models S$ for the sentence $S$
being logically implied by sentence $T$,
where ``sentence'' means closed formula.
The meaning of the implication is that $S$ is true in all models of $T$.
The denotations just defined allow logical implication to be generalized
to apply to formulas that have free variables \cite{clrk91}:

\begin{definition}\label{def:models}
Let $S$ and $T$ be formulas and let $Z$ be the set of variables
occurring in them.
Then we write $T \models S$ to mean
that in all interpretations $R(Z,S) \subset R(Z,T)$.
Likewise, $T \models R(Z,S) \subset R(Z,T)$
means that $R(Z,S) \subset R(Z,T)$ holds in all models of $T$.
\end{definition}

\section{Review of the CLP scheme}
\label{review}

The \clp\ scheme is based on the observation
that in logic programming the Herbrand base
can be replaced by any of many other semantic domains.
Hence the scheme
has as parameter a tuple
$\langle \Sigma, \cal I, L, T\rangle $, where
$\Sigma$ is a signature, $\cal I$ is a $\Sigma$-structure, $\cal L$ is a
class of $\Sigma$-formulas, and $\cal T$ is a first-order $\Sigma$-theory.
These components play the following roles.
$\Sigma$ determines the relations
and functions that can occur in constraints.
$\cal I$ is the structure over which computations are performed
\footnote{
$\cal I$ is a structure consisting of a set $\D$ of values (the {\em
carrier} of the structure) together with relations and functions over $\D$
as specified by the signature $\Sigma$.
For example, there is a complete ordered field
that has $\R$, the set of real numbers, as carrier.
}
.
$\cal L$ is the class of constraints that can be expressed.
Finally, $\cal T$ axiomatizes properties of $\cal I$.

Derivations in the \clp\ scheme are defined by means of transitions
between states.
A state is defined as a tuple
$\langle G,A,P\rangle $ where the goal statement $G$
is a set of atoms and constraints
and $A$ and $P$ are sets of constraints\footnote{
We will often regard $A$ and $P$ as formulas. Then they are the
conjunctions of the atoms they contain.
}.
Together $A$ and $P$ form the {\em constraint store}.
The constraints in $A$ are called the {\em active constraints}\/; those in
$P$ the {\em passive constraints}.

The query $Q$ corresponds to the initial state
$\langle Q,\emptyset,\emptyset\rangle $.
A successful derivation is one that ends in a state of the form
$\langle \emptyset, A, P\rangle$.

The role of $A$ and $P$ in this formula is to describe the answer to the
query $Q$.
$A \wedge P$ is \clp's generalization of Prolog's answer substitution.
It describes an answer, if consistent.
Such an answer may not be useful,
as $P$ may still represent a difficult computational problem.
All that the derivation has done is to reduce
the program atoms to constraint atoms, directly or indirectly
via program atoms.
Derivations also transfer as much as possible
the computational burden of the passive constraints $P$
to the easily solvable active constraints $A$.

\subsection{Operational semantics}
A derivation is a sequence of states such that each next state is
obtained from the previous one by a {\em transition}.
There are four transitions,
$\subr$,
$\subc$,
$\subi$, and
$\subs$
:

\paragraph{1. The resolution transition $\subr\ $:}

$\langle G \cup \{a\},A,P\rangle  \subr
     \langle G \cup B, A, P \cup \{s_1=t_1,\ldots,s_n=t_n\}\rangle$
if $a$ is the atom selected out of $G \cup \{a\}$
by the computation rule, 
$h \leftarrow B$ is a rule of $\cal P$, renamed to new variables,
and $h=p(t_1,\ldots,t_n)$ and $a=p(s_1,\ldots,s_n)$.

$\langle G \cup \{a\}, A, P\rangle  \subr \fail$
is the transition that applies
if $a$ is the atom selected by the computation rule,
and, for every rule $h \leftarrow B$ in $\cal P$, $h$ and
$a$ have different predicate symbols.

\paragraph{2. The constraint transfer transition $\subc\ $:}

$\langle G \cup \{c\}, A, P\rangle  \subc
		    \langle G, A, P \cup \{c\}\rangle$
if constraint $c$ is selected by the computation rule.

\paragraph{3. The constraint store management transition $\subi\ $:}

$\langle G, A, P\rangle  \subi \langle G, A^\prime,
P^\prime\rangle$\\
if $\langle A^\prime, P^\prime\rangle = \infer(A,P)$.

\paragraph{4. The consistency test transition $\subs\ $:}

$\langle G, A, P\rangle  \subs \langle G, A, P\rangle$
if $A$ is consistent;\\
$\langle G, A, P\rangle  \subs \fail$ otherwise.

\subsection{Logic semantics}
For the logic semantics of the \clp\ scheme we follow \cite{clrk91}.

\begin{theorem}[soundness]\label{thm:soundness}
Whenever we have a successful derivation
from query $Q$ resulting in
$P$ and $A$ as passive and active constraints we have
${\cal P,T} \models R(\exists (P\wedge A)) \subset R(Q)$,
where the quantification is over the free variables in $P\wedge A$
that do not occur free in $Q$.
Note Definition~\ref{def:models} for ``$\models$''.
\end{theorem}

\begin{theorem}[completeness]\label{thm:completeness}
Let $Q$ be a query with variables $\Vars_Q$.
If ${\cal P,T} \models R(\Vars_Q, \Gamma) \subset R(Q)$
for a constraint atom $\Gamma$,
then there are $k$ successful derivations from $Q$
with answer constraints $\Gamma_1,\ldots,\Gamma_k$ such that
${\cal T} \models R(\Vars_Q, \Gamma)
     \subset R(\Vars_Q, \Gamma_1) \cup \cdots \cup R(\Vars_Q, \Gamma_k)$.
\end{theorem}
For credits see \cite{clrk91}.

\section{Constraint Satisfaction Problems}
\label{sec:csp}

Constraint Satisfaction Problems (\csps) can be defined as a framework
to cover a variety of specific situations, each exploiting an algorithmic
opportunity. For example, the \csp\ framework can be instantiated to
graph-colouring problems exploiting an efficient algorithm for the
all-different constraint based on matching in bipartite graphs. It
can also be instantiated to the solution of arithmetical constraints
over real-valued variables using efficient algorithms and hardware for
floating-point intervals. It is for this latter instantiation that we are
interested in \csps. But before describing it, first the general framework.

\subsection{CSPs according to Apt}

K. Apt was early in recognizing \cite{apt98a}
that \csps\ can be defined rigorously,
yet in such a way as to be widely applicable.
The following definition is distilled from \cite{apt98a,aptBook2003},
and uses his notation.
\begin{definition}\label{def:apt}
A \csp\ $\langle \mathcal{X},\mathcal{D},\mathcal{C}\rangle$
consists of a sequence
$\mathcal{X}  = \langle x_1, \ldots, x_n \rangle$
of variables,
a sequence
$\mathcal{D}  = \langle D_1, \ldots, D_n \rangle$
of sets called \emph{domains},
and a set
$\mathcal{C}  = \{ c_1, \ldots, c_k \}$
of constraints.
Each constraint is a constraint on a subsequence of $\mathcal{X}$.
An $n$-tuple 
$\langle d_1, \ldots, d_n \rangle \in D_1 \times \cdots \times D_n$
is a solution to 
$\langle \mathcal{X},\mathcal{D},\mathcal{C}\rangle$
iff for every $c \in \mathcal{C}$ on a sequence of variables
$\langle x_{i_1}, \ldots, x_{i_m} \rangle$ from $\mathcal{X}$
we have
$\langle d_{i_1}, \ldots, d_{i_m} \rangle \in c$.
\end{definition}
In this definition $\mathcal{X}$ is probably intended to consist
of $n$ \emph{different} variables.
Once that condition is assumed,
$\mathcal{X}$ need not be a sequence,
but can be a set without further qualification.

\begin{example}\label{ex:aptSol}
To see Apt's definition at work, consider the following example.
$\mathcal{X} = \langle x_1,x_2,x_3,x_4\rangle$,
$\mathcal{D} =
\langle \mathcal{N},\mathcal{N},\mathcal{N},\mathcal{N}\rangle$,
and
$\mathcal{C} = \{c_1,c_2\}$.
Constraints $c_1$ and $c_2$ are on
$\langle x_2,x_2,x_1\rangle$ and
$\langle x_3,x_4,x_1\rangle$, respectively.
To determine some of the solutions
we construct Table~\ref{tab:aptEx}.
\end{example}

\begin{table}
\begin{tabular}{c|c|c|c|c}

        & $x_1$ & $x_2$ & $x_3$ & $x_4$ \\ 
\hline 
$c_1$   &   0    &  0   &       &        \\
        &   2    &  1   &       &        \\
        &   4    &  2   &       &        \\
        &   6    &  3   &       &        \\
        &$\ldots$&$\ldots$&     &        \\
$c_2$   &   0    &      &   0   &   0    \\
        &   1    &      &   0   &   1    \\
        &   1    &      &   1   &   0    \\
        &   2    &      &   0   &   2    \\
        &   2    &      &   1   &   1    \\
        &   2    &      &   2   &   0    \\
        &        &$\ldots$&$\ldots$&$\ldots$\\
solutions &   0    &  0   &   0   &  0    \\
        &   2    &   1  &   0   &   2    \\
        &\ldots&$\ldots$&$\ldots$&$\ldots$\\
\end{tabular}
\caption{Table for Example~\ref{ex:aptSol}: illustrating solution
         according to Apt's definition.}
\label{tab:aptEx}
\end{table}

The table for $c_1$ is constructed according to the rule
$x_2+x_2=x_1$;
for $c_2$ the rule is
$x_3+x_4=x_1$.
In Definition~\ref{def:apt} a constraint
remains a black box:
there is no opportunity to specify a rule
according to which the tuples are constructed.
This omission can be a disadvantage,
as is seen in the important type of discrete \csp\
that can be viewed as a graph-colouring problem.
In practical applications such \csps\ have a small
domain, consisting of the ``colours''.
At the same time they have a large number of variables
and a large number of constraints,
both numbers running in the thousands.
Yet all these constraints have an important property
in common: they derive from the ``all different''
constraint that requires that no two of their arguments
have the same value.

The remedy for this problem was prepared by Definition~\ref{def:relOn},
which is used in our alternative Definition~\ref{def:cspVE} for \csp.
If the definition of \csp\ included a language for expressing constraints,
then these expressions would clarify the connection between $c_1$ and $c_2$.
For example,
$sum(x_2,x_2,x_1)$ would be a good expression for $c_1$ and
$sum(x_3,x_4,x_1)$ for $c_2$.

Predicate logic is a potential candidate
for a formal constraint language.
To realize this potential we modify Apt's definition
to obtain the definition given in the following section.
To be able to interface the solving algorithm for \csps\
with the \clpdc\ scheme, we modify the algorithm also.
In the section after that we define how predicate logic can be
used as the constraint language.

\subsection{A modified definition of CSPs}
\label{sec:cspMod}

\begin{definition}\label{def:cspVE}
A Constraint Satisfaction Problem (\csp)
consists of a finite set
$\Vars = \{ x_1,\ldots,x_n\}$ of variables,
a finite set
$\C = \{c_1,\ldots,c_m\}$ of constraints,
each of which is a relation over a sequence of elements of $\Vars$
in the sense of Definition~\ref{def:relOn}.
With each variable $x_i$ is associated a universe $\D_i$,
which is the set of values that $x_i$ can assume.
A \emph{solution} of a \csp\ is an assignment
to each variable $x_i$ of an element of $\D_i$
such that each constraint in $\C$ is satisfied.
\end{definition}

Apparently, the solution set of a \csp\ with set $\mathcal{X}$
of variables is a relation on
$\mathcal{X}$ in the sense of Definition~\ref{def:relOn}.
A compact characterization of the solution set
can be given as follows.

\begin{lemma}
The solution set equals
$ c_1 \Join \cdots \Join c_m $
where $\Join$ is as in Definition~\ref{def:join}. 
\end{lemma}

For certain \csps\ it is practical to enumerate the solutions.
In other cases the solution set, though finite,
is too large to be enumerated.
And it may be the case that the solution set is uncountable;
moreover its individual solution tuples may consist
of reals that are not computer-representable.

Thus it is often necessary to approximate the solution set.
A convenient form is that of a Cartesian product
$D_1 \times \cdots \times D_n$ that is contained in
$\D_1 \times \cdots \times \D_n$.
Such an approximation has the property that $x_i \not \in D_i$
for any $i \in \{1,\ldots,n\}$ ensures that
$\langle x_1,\ldots,x_n\rangle$ is not a solution.
Making $D_1,\ldots,D_n$ as small as possible gives us as
much information about the solution set as is possible
for approximations of this form.

$D_i$ is called the \emph{domain} for $x_i$ for
$i \in \{1,\ldots,n\}$.
We need to ensure that the subsets $D_i$ of $\D_i$
are computer-representable.
This may not be a restriction when $\D_i$ is finite and small.
It is when $\D_i = \R$.
In general we require that the subsets of $\D_i$ that are allowed
as $D_i$ include $\D_i$ itself and are closed under intersection.
We call such subsets a \emph{domain system}.

\begin{lemma}\label{lem:closed}
Given sets 
$\D_1,\ldots,\D_n$, each with a domain system
and $S \subset \D_1 \times \cdots \times \D_n$.
There is a unique least Cartesian product of domain system elements
containing $S$.
\end{lemma}

\begin{definition}
Given sets 
$\D_1,\ldots,\D_n$, each with a domain system
and $S \subset \D_1 \times \cdots \times \D_n$.
The least Cartesian product of domain system elements containing $S$,
which exists according to Lemma~\ref{lem:closed},
is denoted $\Box S$.
\end{definition}

With each constraint there is associated
a \emph{domain reduction operation} (\dro),
which is intended to reduce
the domains of one or more variables occurring in the constraint.

\begin{definition}\label{def:dro}
Given a \csp\ and a relation
$\rho \subset \D_{i_1} \times \cdots \times \D_{i_k}$. 
Let constraint $c$ be a relation $\rho$ on
$\langle x_{i_1}, \ldots, x_{i_k}  \rangle$
A \emph{domain reduction operation} (\dro\ ) for $c$ is a function
that maps Cartesian products
$
D_{i_1} \times \cdots \times D_{i_k}
\subset
\D_{i_1} \times \cdots \times \D_{i_k}
$
to Cartesian products of the same type.
The map of the function is given by
$
D_{i_1} \times \cdots \times D_{i_k}
\mapsto
D'_{i_1} \times \cdots \times D'_{i_k}
$
where
$
D'_{i_1} \times \cdots \times D'_{i_k}
$
satisfies
$$
\Box((D_{i_1} \times \cdots \times D_{i_k}) \cap \rho)
\subset
D'_{i_1} \times \cdots \times D'_{i_k}
\subset
D_{i_1} \times \cdots \times D_{i_k}.
$$
If the left inclusion is equality,
then we call the \dro\ a \emph{strong} one.
\end{definition}

This operation was introduced by \cite{bnldr97}
under the name ``narrowing''.
The intended application had intervals for the domains,
hence the name.

Note that domains are reduced only by removing non-solutions.
As one can see, \dros\ are \emph{contracting}\/:
if they do not succeed in removing anything,
they leave the domains unchanged. 
Strong \dros\ are \emph{idempotent}\/:
multiple successive applications of the same \dro\ have the same
effect as a single application.

Success of the constraint satisfaction method of solving problems
depends on finding efficiently executable strong \dros.

\subsubsection{Constraint Propagation}

\begin{definition}
A \emph{computation state} of a \csp\
is $D_1 \times \cdots \times D_n$
where $D_i \subseteq \D_i$ is a domain and is associated with $x_i$,
for $i = 1,\ldots,n$. \\
A \emph{computation} of a \csp\ is a sequence of computation states
in which each (after the initial one) is obtained
from the previous one by applying the \dro\ of one of the constraints.\\
The \emph{limit} of a computation is the intersection
of its states. \\
A \emph{fair computation} of a \csp\ is a computation
in which each of the constraints is represented
by its \dro\ infinitely many times.
\end{definition}
Fair computations have infinite length.
However, no change occurs from a certain point onward
(domain systems have a finite number of sets).
By the idempotence of strong \dros,
this is detectable by algorithms
that generate fair computations,
so that they can terminate accordingly.

\begin{theorem}\cite{aptEssence}
The limit of a fair computation of a \csp\
is equal to the intersection of the initial state
of the computation with the greatest fixpoint common to all \dros.
\end{theorem}
%

For a given \csp\ the intersection of the states of any fair
computation only depends on the initial state.
It is therefore independent of the computation itself.
Apparently the \csp\ maps the set of Cartesian products
to itself. It is a contracting, idempotent mapping.

\begin{lemma}
\label{thm:soundCSP}
Let $D$ be the initial state of a fair computation of a \csp.
Then the limit of the fair computation contains the intersection
of $D$ with the solution set.
\end{lemma}

\begin{definition}
The transition from the initial state of a computation
to the limit of that computation is called \emph{constraint propagation}.
\end{definition}

The reason for the name is that the effect
of a \dro\ application on a domain may cause
a subsequent \dro\ applications to reduce other domains.

\subsubsection{Enumeration}

Constraint propagation only goes part way toward solving a \csp:
it results in a single Cartesian product containing all solutions.
In general this single Cartesian product needs to be split up
to give more information about any solutions
that might be contained in it.
This is what enumeration does.

Before a more precise definition,
let us sketch the solving process
by means of the \csp\ arising from a graph-colouring problem.
In case constraint propagation yields
an empty domain in the computation state,
the solving process is over: absence of solutions has been proved.
Suppose the resulting computation state
does not have an empty domain.
We only know that any solutions that may exist
are elements of the Cartesian product of the domains.
If all domains are singletons,
then the corresponding tuple is a solution.
If not, one enumerates a domain
with more than one element (say, the smallest such).
In turn, for each element in that domain,
one assumes it as the value of the variable concerned
and leaves the other domains unchanged.
To the smaller CSP thus obtained, one applies constraint propagation.
This may, in turn, require enumeration; and so on.

To make the idea applicable to the case
where there are infinite domains,
we split a domain instead of enumerating it.
Then it works as above if the domains are countable.

To split an uncountable domain,
then we need the property that the domain system is finite.
Splitting is restricted to producing results
that belong to the domain system.
This implies that only a finite number of splits
are possible.
In case of an uncountable domain
it is not in general possible to identify solutions.

Enumeration yields tuples consisting of domains
that are as small as the domain system allows
that together contain all solutions, if any exist.
And of course solving the CSP results in eliminating
almost all of the Cartesian product of the initial domains
as not containing any solutions.

\paragraph{Enumeration algorithm}

\begin{tabbing}
To enumerate computation state $S$:\\
xxx \=    \kill
\> If a domain is empty, then halt.\\
\> If \= one of the domains is a singleton, \\
\>    \> then substitute the element as value of the corresponding variable \\
\>    \> and construct the computation state $S'$
         with that variable eliminated. \\
\>    \> Enumerate $S'$.\\
\> else \= \kill 
\> else \>  split a domain $d$ into domain system elements
            $d_0$ and $d_1$. \\
\>      \>  Construct computation states $S_i$ by replacing $d$ in $S$
            by $d_i$, for $i = 0,1$. \\
\>      \> Enumerate $S_0$; Enumerate $S_1$.
\end{tabbing}

Often too many enumeration results are generated.
Sometimes the domain system comes with a suitable notion of adjacency
so that adjacent enumeration results can be consolidated into a single one.
Such a consolidation may trigger further consolidations.

\section{The Domain Constraint subscheme of the CLP Scheme}
\label{domain-constraints}

The \clp\ scheme is open-ended:
it is basically a scheme for using Horn-clause rules to generate a multitude of constraint-satisfaction problems.
The parameters of the scheme allow
a great variety of useful algorithms
and of data-types for these to act on.
A first step in reducing the vast variety of options
is \clpdc, the domain-constraint subscheme of the \clp\ scheme.
We define \clpdc\ by visiting
first the parameters $\langle \Sigma, \cal I, L, T\rangle $,
and then the transitions of the \clp\ scheme.

\subsection{The parameters}
$\Sigma$: Some domains are such that individual elements
may not be representable in a computer,
if only because there are infinitely many of them.
Satisfactory results can still be obtained
by designating a finite set of subsets of the domain
that are computer-representable.
To accommodate these the signature $\Sigma$
includes a unary \emph{representability predicate}
for each of these subsets.

$\cal I$: the domain component $\D$ of the $\Sigma$-structure $\cal I$
has to admit a domain system: a finite set of subsets of $\D$
that includes $\D$ and is closed under intersection.

$\cal L$: the language of constraints
consists of conjunctions of atomic formulas.

$\cal T$: the theory giving the semantics of the constraints
links unary representability predicates to representable subsets
of $\D$.
This is done in part by clauses describing the effects of the \dros.
In Definition~\ref{def:dro}
let the constraint $c$ be $r(x_{i_1}, \ldots, x_{i_k})$.
Then the clauses describing the \dro\ of $c$ are
\begin{equation}\label{eq:droClause}
d'_j(x_{i_j})
 \leftarrow 
d_1(x_{i_1}), \ldots,
d_k(x_{i_k}),
r(x_{i_1}, \ldots, x_{i_k})
\end{equation}
for $j = 1, \ldots, k$.
Further details depend on the instance concerned of \clpdc.
The idea of expressing the action of a \dro\
in the form of an inference rule is due to \cite{apt98a}.
This is closely related to the inclusion of an implication
like the one above in a theory.

\subsection{The transitions}
The $\subr\ $ and $\subc\ $ transitions:
These only serve to transform goal atoms into constraint atoms,
and are needed unchanged in the \clpdc\ subscheme.

\paragraph{}
The $\subi\ $ transition:
In the \clp\ scheme this transition
is intended to accommodate any inference
that transfers the burden of constraint
from the passive constraints $P$
to the efficiently solvable active constraints $A$.
In the \clpdc\ subscheme such inference is restricted
to those forms that leave $P$ unchanged:
the information contained in them is only used
to strengthen the active constraints $A$.
Moreover, $A$ is restricted to the form
$\{d_1(x_1),\ldots,d_n(x_n)\}$
where each variable in the passive constraint $P$
occurs exactly once and where $\{d_1,\ldots,d_n\}$
are unary representability predicates.

As $P$ is the unchanging conjunction of the constraints,
we refer to it as $C$ in the \clpdc\ subscheme.
As $A = \{d_1(X_1),\ldots,d_n(X_n)\}$
only states of each of the variables
that it belongs to a certain domain
we refer to it as $D$ in the \clpdc\ subscheme.
As a result of these renamings
we have a close relationship between \csp\ and \clpdc:
$C$ and $D$ in \csp\ and in \clpdc\ are counterparts of each other.
As a result of these restrictions and renamings,
the constraint store management transition becomes
$\langle G, D, C\rangle  \subi \langle G, D', C\rangle$
if $\langle D', C\rangle = \infer(D,C)$.

The \emph{infer} operation is performed by setting up a \csp\
with an initial state and determining the limit of the fair
computations from the initial state.
This limit is then the $D'$ in $\langle D', C\rangle = \infer(D,C)$.

The \csp\ that implements \emph{infer} in this way
has the following components.
\begin{enumerate}
\item
The variables are those that occur in the passive constraint $C$.
\item
The universes over which the variables range
are equal to each other and to $\D$.
\item
If a constraint atom $c_j$ of \clpdc\ is $r(x_{i_1}, \ldots, x_{i_{k_j}})$,
then the corresponding constraint of the \csp\ is
$\rho$, the meaning of $r$, on
$\langle x_{i_1}, \ldots, x_{i_{k_j}}\rangle$,
with ``on'' as in Definition~\ref{def:relOn}.
\item
If the active constraint is $\{d_1(x_1),\ldots,d_n(x_n)\}$,
then the initial computation state in the \csp\ is
$D_1 \times \cdots \times D_n$
with $D_i = \{ x \in \D \mid d_i(x) \}$, for $i = 1,\ldots,n$.
\end{enumerate}

In the \csp\ thus obtained a fair computation is constructed
with limit
$D'_1 \times \cdots \times D'_n$.
These domains are then used to determine the active constraint
$D' = \{d'_1(x_1),\ldots,d'_n(x_n)\}$,
where the $d_i$ are obtained from
$D'_i = \{ x \in \D \mid d'_i(x) \}$, for $i = 1,\ldots,n$.

In this way \csp\ computations can be used in \clpdc.

\paragraph{}
The $\subs\ $ transition:
In the \clp\ scheme this transition checks as best as it can
whether $P \wedge A$ is consistent.
In the \clpdc\ subscheme no attempt is made to check $C$ for consistency.
It does this only for $D = d_1(x_1) \wedge \cdots \wedge d_n(x)$
and this is simply a check
whether any of the $d_i$ is the predicate for the empty subset of $\D$.

\begin{lemma}
The existence of a successful derivation implies that
${\cal C,T} \models R(C) \subset R(Q)$.
\end{lemma}
\begin{proof}
By theorem~\ref{thm:soundness} we have 
${\cal C,T} \models R(C \wedge D) \subset R(Q)$
and we have ${\cal C,T} \models R(C) \subset R(D)$.
\end{proof}

\section{Interval constraints} \label{sec:intervConstr}

We have used in Section~\ref{review}
the \clp\ scheme as starting point.
To establish the direction in which to proceed,
we identified in Section~\ref{sec:csp}
a desirable point outside of logic programming:
the \csp\ paradigm.
Here are to be found useful algorithms
for computational tasks of interest.
These range from the ``most discrete'' such as graph colouring
to the ``most continuous'' such as solving non-linear
equalities and inequalities over the reals.

After thus establishing a line along which to travel,
we went back in Section~\ref{domain-constraints}
to establish a subscheme of the \clp\ scheme,
that of the domain constraints,
to emulate within logic the main features of the \csp\
paradigm.

It is now time to declare our main interest:
real valued variables rather than discrete ones.
It so happens that there is a real-variable specialization
of the \csp\ paradigm, interval constraints,
and it will be useful to take an excursion from logic again
and review this next.

We are interested in \csps\ with the following characteristics.
The variables range over the reals;
that is, all universes
$\D_1,
\ldots,
\D_n$
are equal to the set $\R$ of reals. 
The domain system is that of the floating-point intervals.
The constraints include
the binary $\leq$ and the ternary $sum$ and $prod$.
The reason is that these have strong \dros\
that are efficiently computable.
Strong \dros\ are also available
for $=$, $max$, $abs$, and for rational powers.
For the constraints corresponding to the transcendental functions
\dros\ are available that are idempotent, but not strong.
The definition of ``strong'' requires them
to be the least floating-point box
containing the intersection of the relation with the argument box.
That the \dro\ is not strong
has to do with the difficulty of bounding these function
values between adjacent floating-point numbers.
But \dros\ closely approximating this ideal
are used in some systems \cite{vhlmyd97}.

Let us consider an example of a \dro\
for use with real-valued variables constrained
by the relation
$$
sum = \{\langle x,y,z \rangle \in \Rea  \mid x+y = z\}.
$$
Suppose the domains for
$x$,
$y$, and
$z$
are
$[0,2]$,
$[0,2]$, and
$[3,5]$.
Clearly, neither $x$ nor $y$ can be close to $0$,
nor can $z$ be close to $5$.
Accordingly, when this \dro\ is applied, 
these intervals are reduced to
$[1,2]$,
$[1,2]$, and
$[3,4]$.

The numbers $1$ and $4$ arise by computing
$3-2$ and $2+2$.
Here no rounding errors were made.
This is exceptional.
Let us now consider the case in which
the initial intervals are scaled down by a factor of ten to
$[0.0,0.2^{+}]$,
$[0.0,0.2^{+}]$, and
$[0.3^{-},0.5^{+}]$.
Here $0.2^{+}$ is the least floating-point number
not less than $0.2$,
and similarly for the other superscripts.
Now the corresponding operations $0.3^{-}-0.2^{+}$ and $0.2^{+}+0.2^{+}$
do incur rounding errors.
$0.3^{-}-0.2^{+}$ is evaluated to a floating-point number
we shall name $0.1^{- -}$;
similarly, $0.2^{+}+0.2^{+}$ is evaluated to $0.4^{+ +}$,
so that the \dro\ gives the intervals
$[0.1^{--},0.2^{+}]$,
$[0.1^{--},0.2^{+}]$, and
$[0.3^{-},0.4^{+ +}]$
for $x$, $y$, and $z$, respectively.
Here $0.1^{--}$ may equal $0.1^{-}$ or $(0.1^{-})^{-}$.
The decimal equivalents
of the binary floating-point numbers computed here are so lengthy
that users are neither willing to write nor to read them,
so that further containment precautions
are called for on (decimal) input and output.

In this way single arithmetic operations
find their counterpart in interval constraints.
To give an idea of how the arbitrarily complex arithmetic expressions
in nonlinear equalities and inequalities
are translated to interval constraints
consider the equation $1/x+1/y=1/z$ relating the resistance $z$
of two resistors in parallel with resistances $x$ and $y$.
Constraint processing is not directly applicable when,
as we assume here,
we only have \dros\ for $sum$ and $inv$.
We therefore convert the equation to the equivalent form
$$
\exists u,v,w \in \R.\;
inv(x,u) \wedge inv(y,v) \wedge inv(z,w) \wedge sum(u,v,w). 
$$
Accordingly, the equation is translated to a CSP with
$\Vars = \langle x,y,z,u,v,w \rangle$
and
$\C = \{inv(x,u),inv(y,v),inv(z,w),sum(u,v,w)\}.$

In numerical CSPs we can conclude,
according to Theorem~\ref{thm:soundCSP}, that the solution set
is empty when the limit of the computation is empty.
However, a nonempty limit can still coexist with an empty solution
set.

It is possible to develop \dros\ for complex expressions
such as $1/x+1/y=1/z$
\cite{bmcv94}.
It is useful to know that this paper is antedated
by the technical report version of \cite{bnldr97}.

\section{CLP/NCSP: the CLP/DC subscheme with a numerical CSP}
\label{sec:clpncsp}

In Section~\ref{domain-constraints} we described
how the open-ended \clp\ scheme is narrowed down
to the domain-constraints subscheme \clpdc.
In this section we take a step further in this direction
to obtain a subscheme suitable for numerical computation.
We do this by following the specification
in Section~\ref{domain-constraints}.

\subsection{The hierarchy of theories}
$\Sigma$: 
The signature contains the language elements needed for the usual
theory of the real numbers:
constants including $0$ and $1$; the unary function symbol $-$;
the binary function symbols
$+$,
$-$,
$*$, and
$/$;
the binary predicates $\leq$ and $\geq$. 
To these we add:
\begin{itemize}
\item
A unary representability predicate
$d_{a,b}$ for every floating-point interval in a given floating-point
number system.
For the IEEE-standard double-length floating-point numbers
this means in the order of $2^{127}$ unary predicates.
Not a mathematically elegant signature, but a finite one.
\item
Ternary predicates $sum$ and $prod$.
\end{itemize}

\paragraph{}
$\cal I$: 
the domain component $\D$ of the $\Sigma$-structure $\cal I$
is the set $\R$ of real numbers.
The domain system consists of the floating-point intervals,
which are sets of reals.
The floating-point intervals include $\R$ itself
and are closed under intersection,
so include the empty interval.

$\cal T$: to the axioms of the usual theory of the reals
we add:
\begin{eqnarray*}
\forall x.\;[d_{-\infty,b}(x) & \leftrightarrow &
      x \leq b] \mbox{ for every floating-point number } b
\\ 
\forall x.\;[d_{a,b}(x)       & \leftrightarrow &
      a \leq x, x \leq b] \mbox{ for every pair of flpt numbers such that }
      a \leq b
\\ 
\forall x.\;[d_{a,+\infty}(x) & \leftrightarrow &
      a \leq x] \mbox{ for every floating-point number } a
\\ 
\forall x,y,z.\;[sum(x,y,z) &\leftrightarrow & x+y=z]
\\ 
\end{eqnarray*}
We refer to the resulting theory as ${\cal T}_1$.
The only difference with the usual axiomatization of the reals
is that meanings are established for the newly introduced
predicates.

The effect of the \dro\ of a constraint is described in clauses
as in equation~\ref{eq:droClause} in Section~\ref{domain-constraints}.
For each of the atomic constraints in the passive constraint $C$
this causes clauses to be added to ${\cal T}_1$.
We call the resulting theory ${\cal T}_2$.

\begin{theorem}\label{thm:soundNCSP}
Let $C$ be the passive constraint,
let $D$ be the initial active constraint,
and let $D'_1,\ldots,D'_m$ be the active constraints
corresponding to the results of a \csp\ enumeration
starting with initial constraint corresponding to $D$
and constraints corresponding to $C$.  Then
${\cal T}_2 \models R(C\wedge D) \subset [R(D'_1)\cup \cdots \cup R(D'_m)]$.
\end{theorem}

It would be more convincing if we could assert that
${\cal T} \models R(C\wedge D) \subset [R(D'_1)\cup \cdots \cup R(D'_m)]$,
as $\cal T$ is the usual theory of the reals,
without computer-related artifacts.
This is not possible,
as $R(C)$ and $R(D)$ contain constraint predicates 
and these do not occur in $\cal T$.
However, all axioms that are in ${\cal T}_2$
and not in $\cal T$
are logical consequences of $\cal T$.
\begin{proof}
Every application of a \dro\ corresponds to an inference
with one of the rules in ${\cal T}_2$ of the form
of Equation~(\ref{eq:droClause}).
\end{proof}

It is now time to look at examples of what we can do with the tools
developed so far.
The first two examples concern
a polynomial in a single real variable
and represent it by a term $p$ in the variable $x$.
In these examples the problem is stated in a single constraint,
so only uses a part of the \clp\ paradigm.
The third example is a toy design problem.
Here the \clp\ paradigm is fully exercised:
multiple derivations are generated,
each of which is potentially a significant numerical \csp.

\subsection{Semantics of solving numerical inequalities}

Consider the problem of determining
where the given polynomial is non-positive.
This corresponds to the constraint $p \leq 0$.
In ${\cal T}_2$ we can translate $p \leq 0$ to a set $C$ of constraints.
For example, if $p$ is $x*(x-2)$ we have in ${\cal T}_2$
$$
\forall x [
  x * (x-2) = 0 \leftrightarrow
  \exists v,w.\; sum(v,2,x) \wedge prod(x,v,w) \wedge w \leq 0]
$$
so that we have the constraint $\Gamma$ equal to
$\{sum(v,2,x) \wedge prod(x,v,w) \wedge w \leq 0\}$.
A highly complex $p$ will give rise to a $C$ with many atoms and many variables.

Soundness (theorem~\ref{thm:soundNCSP}) implies
that the active constraints in the answer constraint
for this problem \clpncsp\ contains all intervals in which
$p$ is zero or negative.
Completeness implies
that whenever we have for a constraint $\Gamma$ that
\begin{equation}\label{eq:ineqCond}
{\cal T}_2 \models R(\Gamma) \subset R(p \leq 0)
\end{equation}
there are $m > 0$ derivations ending in anwer constraints
$\Gamma_1,\ldots,\Gamma_m$ such that
${\cal T}_2 \models R(\Gamma)
  \subset R(\Gamma_1) \cup \cdots \cup R(\Gamma_m)
$.
We cannot replace Equation~(\ref{eq:ineqCond})
by ${\cal T}_2 \models R(p \leq 0) \neq \emptyset$.
This would be reducible to the problem of deciding equality between two reals,
a problem shown to be unsolvable \cite{aberth98}.

\subsection{Semantics of equation solving}

A well-known numerical problem
that can present computational difficulties
is the one of determining $R(p = 0)$.

Theorem~\ref{thm:soundNCSP} shows
that the active constraints in the answer constraint
for this problem \clpncsp\ contain all zeroes of the polynomial.
It also shows that in case of finite failure the polynomial has no zeroes.
The possibility remains that finite failure does not occur,
yet there are no zeroes.
This is unavoidable.
The problem of deciding whether
${\cal T}_2 \models R(p = 0) = \emptyset$
reduces again to the problem of deciding equality between two reals.
The best we can hope for is attained here:
showing emptiness or finding small intervals
in which all solutions, if any, are contained.

Completeness (Theorem~\ref{thm:completeness})
has nothing to say about this problem:
it is rare for a polynomial $p$ to make
${\cal T}_2 \models R(\Gamma) \subset R(p = 0)$
true for non-empty $R(\Gamma)$.
With respect to ${\cal T}_2$ the set $R(p = 0)$
is a finite set of reals,
and it is rare for these to be a floating-point number.
For most $p$, the least $R(\Gamma)$ containing containing
any root of it is an interval of positive width. 

Conventional numerical computation produces single floating-point
numbers that are intended to be near a solution,
and mostly are.
Sometimes they are not, and one cannot tell from the program's output.
Interval arithmetic and numerical \csps\ improve on this by
returning intervals that contain the solutions, if any,
and by failing to return any intervals in which it is certain that
no solutions exist.
\clpncsp\ improves on this by giving a logic semantics,
of which Theorem~\ref{thm:soundNCSP} is an example.
However, interval arithmetic and interval constraints are limited
in that they only solve a single \csp.
A more important advantage of \clpncsp\ is that,
in addition to solving \csps,
it automates the generation of the multiple \csps\
that are often required in scheduling and in engineering design.
We close by giving an example of this mode of operation.

\subsection{A toy example in CLP/NCSP}
Consider an electrical network
in which resistors are connected to each other.
The network as a whole has a certain resistance.
We have available twelve resistors;
three each of 100, 150, 250, and 500 ohms.
From this inventory we are to build a network
that has a specified resistance
so that it can function as part in a larger apparatus.
Fortunately there is a certain latitude:
the resistance of the resulting network
has to lie between 115 and 120 ohms. 
\emph{The structure of the network is not given.}
This is a design problem in addition to
being a computational problem.

Even with the dozen components given in this problem
there is a large number of ways in which they can be connected.
We can nest parallel networks inside a series network,
or the other way around,
to several levels deep.
Evaluation of each such combination
requires a non-negligible amount of computation
involving real-valued variables.
The search space is sizable,
hence the importance of constraint propagation
to eliminate most of it.

Let us imagine for \clpncsp\ a Prolog-like syntax.
Please do not be misled by the {\tt type writer-like}
font into believing in an implementation: none exists.
The figures given in the example are for illustration only
and are chosen to be merely plausible.

According to \clp\ bodies of clauses contain
both constraint atoms and program atoms.
We separate them with a semicolon:
the constraints, if any, come first.
Instead of writing $d_{a,b}(x)$ for the domain constraints,
we write for ease of typing
{\tt <a|X|b>} in the style of Dirac's bra and ket notation.
When \verb+a+ is infinite,
we write \verb+-inf+; this is a single mnemonic identifier,
denoting that particular floating-point value.
Similarly for \verb+b+ and 
\verb|inf| or \verb|+inf|.
We omit constraints like
{\tt <-inf|X|+inf>}, which do not constrain their argument.

The predicate \verb+netw(A,N,B,R,PL)+ asserts that network
represented by
\verb+N+
connects terminals
\verb+A+
and
\verb+B+,
has resistance
\verb+R+,
and has parts list
\verb+PL+.
The term \verb+N+ can be
\verb+at(X)+
for an atomic network, which is in this case a single resistor;
it can be
\verb+ser(N1,N2)+,
for two networks in series, or 
\verb+par(N1,N2)+,
for two networks in parallel.

\begin{center}
\begin{verbatim}
1: netw(A,at(R),B,R,(r150:1).nil)
   :- <149.9|R|150.1>;.
% Similarly for 100, 250, and 500 ohms.
2: netw(A,ser(N1,N2),C,R,PL)
   :- sum(R1,R2,R);
      netw(A,N1,B,R1,PL1), netw(B,N2,C,R2,PL2),
      merge(PL1,PL2,PL).
3: netw(A,par(N1,N2),B,R,PL)
   :- inv(R,RR),inv(R1,RR1),inv(R2,RR2), sum(RR1,RR2,RR);
      netw(A,N1,B,R1,PL1), netw(A,N2,B,R2,PL2),
      merge(PL1,PL2,PL).
\end{verbatim}
\end{center}

Clause 1 says that the network can be atomic,
consisting of a single resistor with resistance \verb+R+,
represented by the term \verb+at(R)+.
Its parts list is a list consisting of a single item
\verb+r150:1+,
being a resistor of nominal value 150 ohms in quantity 
\verb+1+.
The condition of clause 1
states that the actual resistance
\verb+R+,
a real variable, belongs to the interval
\verb+[149.9, 150.1]+,
which expresses the tolerance.
There are similar clauses to represent
the other sizes of resistor that are available.

Clause 2 says the network can be
\verb+ser(N1,N2)+,
the series composition of two networks
\verb+N1+
and
\verb+N2+
of unspecified structure, with resistance
\verb+R+,
that satisfies the constraint for the resistance
of a serial composition of networks:
\verb"sum(R1,R2,R)",
which means that $R_1 + R_2 = R$.

In clause 3 the constraint means
$1/R_1 + 1/R_2 = 1/R$,
which is the constraint for resistances $R_1$ and $R_2$
in parallel giving resistance $R$.

The predicate 
\verb+merge+
is left as a black box. Suffice it to know that
the goal \verb+merge(PL1,PL2,PL)+
merges parts lists
\verb+PL1+
and
\verb+PL2+,
which satisfy the inventory restrictions,
into parts list \verb+PL+
unless the latter does not satisfy the inventory restrictions,
in which case the goal fails.

The query
\begin{verbatim}
:- <149.9|R150|150.1>, ...;
   netw(a, par(at(R150), ser(at(R500), par(at(R100), at(R250)))), b, R, PL).
\end{verbatim}
succeeds without search to an answer
that could include something like \verb+<117.1|R|119.3>+.
The program looks like it has been written
with such queries in mind.
However, as explained below, it also succeeds,
though with some search,
to answer
\begin{verbatim}
:- <115.0|R|120.0>; netw(A,N,B,R,PL).
\end{verbatim}
with
\begin{verbatim}
N = par(at(R150), ser(at(R500), par(at(R100), at(R250))))
\end{verbatim}
and\\ \verb+<117.1|R|119.3>, <149.9|R150|150.1>, ...+.

In response to the latter query \pr\
has \emph{synthesized} a suitable network,
thereby solving the design problem.
It traversed a search space
consisting of multiple \csps\
that was generated by \clp\ derivations.
Many of these derivations were cut short
by failing \csps.

As in the first two examples,
the soundness of Theorem~\ref{thm:soundness}
guarantees for this problem that all networks
that are found have a resistance contained in the required interval.
We noticed that in the case of polynomial roots completeness
has no interesting consequence.
This was true because the problem had the form
of a single constraint with equality.
In the design of a resistor network
there is a goal statement with a program atom.
It gives rise to mutiple derivations.
Completeness (theorem~\ref{thm:completeness})
implies that in case a solution exists,
derivation are generated to cover the given interval
for the network's resistance: all solutions are found.

\section{Concluding remarks}
\pr\ only incorporates numerical \csps\ into the \clp\ scheme.
Other types of \csp\, such as those dealing with finite domains,
can be incorporated in the same way.
In fact, the \clp\ scheme is not restricted to including
special-purpose computation into the logic framework:
it has remedies for those difficulties that prevented Prolog
from being a logic programming language.

Pure Prolog held out the promise of a programming language
with logical-implication semantics.
The impracticality of the occurs check in unification
and of symbolic implementation of numerical computation 
caused standard Prolog to compromise semantics.
In this paper we described a method for including
the power of hardware floating-point arithmetic
without semantical compromise.
We should not lose sight of the fact that the \clp\ scheme
also has a remedy for the other
blemish of standard Prolog: compromised unification.
As Clark \cite{clrk91} showed, the Herbrand Equality Theory,
which requires the occurs check,
is only one possible unification theory for the \clp\ scheme.
It can be replaced by Colmerauer's Rational Tree Equality
Theory, so that we have the prospect of a fully practical
programming language with logic-implication semantics.


\end{document}